\documentclass[journal,twoside,web]{ieeecolor}
\usepackage{generic}
\usepackage{cite}
\usepackage{amsmath,amssymb,amsfonts,cases,mathrsfs,mathtools}
\usepackage{algorithmic,algorithm}
\usepackage{xcolor,graphicx}
\usepackage{hyperref}
\hypersetup{hidelinks=true}
\usepackage{textcomp}
\def\BibTeX{{\rm B\kern-.05em{\sc i\kern-.025em b}\kern-.08em
    T\kern-.1667em\lower.7ex\hbox{E}\kern-.125emX}}
\usepackage{wrapfig,lipsum}
\usepackage{pgfplots}
\usepackage{tikz, tikzscale,calc}
\usepackage{caption}
\usepackage{subcaption}
\usepackage{enumerate}
\usepackage{booktabs}
\usepackage{array}
\usepackage{verbatim} 
\usepackage{multirow}
\newtheorem{lem}{Lemma}
\newtheorem{prop}{Proposition}
\newtheorem{defin}{Definition}
\newtheorem{ass}{Assumption}
\newtheorem{thm}{Theorem}
\newtheorem{remark}{Remark}

\newenvironment{definition}{\vskip 3pt\begin{defin}}{\vskip 3pt
\end{defin}}

\begin{document}
\title{Stability of Open Multi-agent Systems over Dynamic Signed Digraphs}
\author{Pelin Sekercioglu, Angela Fontan, Dimos V. Dimarogonas
\thanks{This work was supported in part by the Wallenberg AI, Autonomous Systems and Software Program (WASP) funded by the Knut and Alice Wallenberg (KAW) Foundation, the Horizon Europe EIC project SymAware (101070802), the ERC LEAFHOUND Project, the Swedish Research Council (VR), and Digital Futures.}
\thanks{P. Sekercioglu, A. Fontan, and D. V. Dimarogonas are with the Division of Decision and Control Systems, KTH Royal Institute of Technology, SE-100 44 Stockholm, Sweden (e-mail: \{pelinse,angfon,dimos\}@kth.se).}}

\maketitle

\begin{abstract}
We address the synchronization problem in open multi-agent systems (OMAS) containing both cooperative and antagonistic interactions. In these systems, agents can join or leave the network over time, and the interaction structure may evolve accordingly. To capture these dynamical structural changes, we represent the network as a switched system interconnected over a dynamic and directed signed graph. Additionally, the network may contain one or multiple leader groups that influence the behavior of the remaining agents. In general, we show that the OMAS exhibit a more general form of synchronization, including trivial consensus, bipartite consensus and containment. Our approach uses the signed edge-based agreement protocol, and constructs strict Lyapunov functions for signed networks described by signed edge-Laplacian matrices containing multiple zero eigenvalues. Numerical simulations validate our theoretical results.
\end{abstract} 

\begin{IEEEkeywords}
Open multi-agent systems, signed networks, signed edge-based agreement protocol, switched systems.
\end{IEEEkeywords}

\section{Introduction}
\label{sec:introduction}
\IEEEPARstart{O}{pen} multi-agent systems (OMAS) are dynamic networks in which agents and edges can be added and removed over time. These systems naturally arise in applications such as social networks, where the size and structure of the graph evolve due to arrivals and departures of the participants \cite{hendrickx2016open,hendrickx2017open}. In such settings, relationships may also change as individuals' interactions evolve through the formation of new friendships, the dissolution of old ones, or shifts in trust or like/dislike over time. OMAS also appear in sensor-based robotic systems, where the topology is adapted based on inter-agent distances to maintain global connectivity of the graph and ensure collision avoidance among nearby agents \cite{restrepo2022consensus,dimarogonas2007decentralized}. 

Early work on OMAS includes \cite{hendrickx2016open}, which studies the consensus problem via a discrete-time gossip algorithm with deterministic agent arrivals and departures. This was later extended to random arrivals in \cite{hendrickx2017open}. Similarly, \cite{abdelrahim2017max} addresses the max-consensus problem using discrete-time gossip-based interactions. Later, in \cite{franceschelli2018proportional, franceschelli2020stability}, new stability definitions for OMAS are introduced. 
Other works on OMAS in social networks and distributed computation include \cite{ vizuete2022resource, varma2018open, oliva2023sum}. In \cite{xue2022stability}, a novel stability framework was proposed to address the consensus problem in OMAS with nonlinear dynamics. The key idea is to model the OMAS as a switched system, where the system switches modes whenever at least one node or edge is added or removed. 
Under an average dwell time condition, the framework guarantees the (practical) stability of the overall system.  
Following the work of \cite{xue2022stability}, \cite{restrepo2022consensus} tackled the problem of formation control for OMAS modeled by first-order systems over undirected graphs and under inter-agent constraints.

In all previous references, it is assumed that all agents cooperate with each other. However, in many real-world scenarios, agents may exhibit antagonistic behaviors. Examples include robotic applications such as herding control \cite{sebastian2022adaptive} and social networks where agents compete \cite{altafini2012_6329411,Fontan2021Role,Fontan2021Signed} and may spread disinformation \cite{csekerciouglu2024distributed}, to name just a few. This work is motivated by such real-world scenarios, particularly addressing (i) changes in the nature of interactions and (ii) the presence of multiple influential entities, referred to as leaders. 
A well-established approach to modeling cooperative and antagonistic interactions in dynamical systems is through the use of \textit{signed graphs} \cite{altafini2012_6329411,Shi2019Dynamics}. In this framework, cooperative interactions are represented by positive edge weights, while antagonistic interactions are modeled with negative weights.

In this paper, we study the synchronization of signed OMAS interconnected over directed signed graphs. We consider systems where new nodes and edges can be added or removed, the nature of the interconnections can change signs, and the direction of information exchange may shift, leading to role changes among agents. Our main contribution is
a novel stability approach for signed OMAS interconnected over dynamic directed signed graphs, potentially with multiple leader nodes, by extending the Lyapunov equation-based analysis to signed edge Laplacians with multiple zero eigenvalues, enabling the construction of strict Lyapunov functions. To the best of our knowledge, this is the first attempt to address OMAS over directed and signed graphs. 
As in \cite{xue2022stability}, we model the signed OMAS using a switched system representation. A first work in this direction is \cite{CDC2025OMAS}, where we addressed the problem of bipartite consensus over OMAS, but only considering undirected signed graphs. In this paper, we address a more general case scenario, where we consider directed signed graphs, and the presence of leaders.

From a technical perspective, our main results build on \cite{csekerciouglu2024distributed} and the framework introduced in \cite{DYNCON-TAC16}, where the dynamical system is decomposed into two interconnected subsystems: the dynamics of the weighted average system and the synchronization errors relative to that average. We extend the latter component to directed signed graphs, described in terms of the edges of the graph. Specifically, we reformulate the synchronization problem as a stability problem for synchronization errors, defined as the difference between the edge states and the weighted edge average. Our analysis is conducted in signed edge-based coordinates. 
Then, building on \cite{ECC20_EP}, which provides a Lyapunov characterization for Laplacians with a single zero eigenvalue, we extend this result to Laplacians with multiple zero eigenvalues. Unlike \cite{csekerciouglu2024distributed}, we specifically consider signed edge Laplacians. Our contributions include the construction of strict Lyapunov functions and the analysis of dynamic signed digraphs, where node additions/removals, cooperation–antagonism switches, and leader–follower role shifts directly impact the system's convergence behavior, setting our work apart from the aforementioned references.

\section{Preliminaries}\label{section3}
\textit{Notation:} $\lvert \cdot \rvert$ denotes the absolute value for scalars, the Euclidean norm for vectors, and the spectral norm for matrices. 
$\mathbb{R}$ is the set of real numbers and $\mathbb{R}_{\geq 0}$ is the nonnegative orthant. $A>0$ ($A\ge 0$) indicates that $A$ is a positive definite (positive semidefinite) matrix. 

\subsection{Signed digraphs}\label{section3a} 
Let $\mathcal{G}_s = (\mathcal{V}, \mathcal{E})$ be a \textit{signed} digraph, where $\mathcal{V} = \{ \nu_1, \nu_2, \dots, \nu_N \}$ is the set of $N$ nodes and $\mathcal{E} \subseteq \mathcal{V} \times \mathcal{V}$ is the set of edges. The edge $\varepsilon_k = ( \nu_j, \nu_i ) \in \mathcal{E}$ of a digraph denotes that the agent $\nu_i$, which is the terminal node (tail of the edge), can obtain information from the agent $\nu_j$, which is the initial node (head of the edge). The adjacency matrix of $\mathcal{G}_s$ is $A := [a_{ij}] \in \mathbb{R}^{N \times N}$, where $a_{ij}\ne 0$ if and only if $(\nu_j,\nu_i)\in \mathcal{E}$. $a_{ij} > 0$ if and only if the edge $(\nu_j, \nu_i)$ has a positive sign, indicating a cooperative relationship, and $a_{ij} < 0$ if and only if the edge $(\nu_j, \nu_i)$ has a negative sign, indicating an antagonistic relationship. In this work, we only consider unweighted digraphs, such that $a_{ij} = \{0, 1, -1\}$, without any self-loops. A signed digraph is said to be digon sign-symmetric if $a_{ij} a_{ji} \ge 0$. It means that, if $(\nu_i, \nu_j)$ and $(\nu_j, \nu_i)$ both exist, the interaction between two interconnected agents always has the same sign in both directions. Then, we make the following assumption:

\begin{ass}\label{standing_ass}
    The signed digraph is unweighted and digon sign-symmetric.
\end{ass}

A directed path is a sequence of distinct adjacent nodes in a digraph. When the nodes of the path are distinct except for its end vertices, the directed path is called a directed cycle. A directed spanning tree is a directed tree subgraph that includes all the nodes of the digraph. In this structure, every agent (node) has a parent node, except for the \textit{root node}, which has no incoming edges and is connected to every other node via directed paths. A tree contains no cycles. A digraph is said to be strongly connected if there exists a directed path between every pair of nodes. A directed graph is weakly connected if replacing all its directed edges with undirected edges results in a connected undirected graph.

In this work, we consider signed digraphs that may contain multiple leader nodes. A leader node is defined as either a root node or a node that is part of a \textit{rooted strongly-connected component (rooted SCC)}. A rooted SCC is a strongly connected subgraph without incoming edges. A \textit{leader group} is either a single root node (representing a single-node leader group) or an entire rooted SCC (representing multiple leader nodes interconnected in a strongly connected subgraph). If the graph contains at least one leader group, the remaining nodes are referred to as followers--- See Figure \ref{leadergroups}.

A signed graph is said to be \textit{structurally balanced} (SB) if it can be split into two disjoint sets of vertices $\mathcal{V}_{1}$ and $\mathcal{V}_{2}$, where $\mathcal{V}_{1} \cup \mathcal{V}_{2} = \mathcal{V}, \mathcal{V}_{1} \cap \mathcal{V}_{2} = \emptyset$ such that for every $\nu_{i}, \nu_{j} \in \mathcal{V}_{p}, p \in \{ 1,\, 2 \}$, if $a_{ij} \geq 0$, while for every $\nu_{i} \in \mathcal{V}_{p}, \nu_{j} \in \mathcal{V}_{q}$, with $p,q \in \{ 1, \, 2 \}, p \neq q$, if $a_{ij} \leq 0$\footnote{Agents in different subsets need not be directly connected: agents belonging to the same subset are cooperative, while agents in different subsets are antagonistic, as determined by the signs of paths connecting them.}. It is {\it structurally unbalanced} (SUB), otherwise \cite{altafini2012_6329411}. The signed Laplacian matrix $L_{s}=[{\ell _{s_{ij}}}] \in \mathbb{R}^{N \times N}$ associated with $\mathcal{G}_s$ is defined as ${\ell _{s_{ij}}} := {\sum\limits_{k \leq {N}} {{ \lvert a_{ik} \rvert }} }$, if ${i = j}$; and ${\ell _{s_{ij}}} := { - {a_{ij}}}$, if ${i \ne j}$ \cite{altafini2012_6329411,Shi2019Dynamics}.
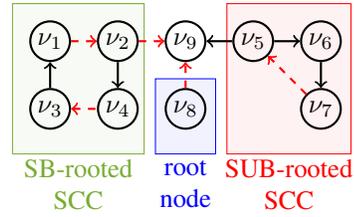
\begin{figure}[!t]\centering
\begin{tikzpicture}[node distance={9mm}, thick,main/.style = {draw, circle}] 
    \definecolor{pastelgreen}{rgb}{0.4660 0.6740 0.1880}
	\tikzset{mynode/.style={circle,draw,minimum size=15pt,inner sep=0pt,thick},}
    \draw[pastelgreen,very thin, fill=pastelgreen!5](-0.5,0.5) rectangle(1.25,-1.5);
    \node[align=center,pastelgreen] at (0.375,-1.9) {SB-rooted\\SCC};
    \draw[red,very thin, fill=red!5](2.35,0.5) rectangle(4,-1.5);
    \node[align=center,red] at (3.175,-1.9) {SUB-rooted\\SCC};
    \draw[blue,very thin, fill=blue!5](1.4,-0.5) rectangle(2.2,-1.5);
    \node[align=center,blue] at (1.8,-1.9) {root\\node};
	\node[mynode] (1) {$\nu_1$}; 
	\node[mynode] (2) [right  of=1] {$\nu_2$};
	\node[mynode] (3) [below of=1] {$\nu_3$};
	\node[mynode] (4) [below of=2] {$\nu_4$};
    \node[mynode] (9) [right of=2] {$\nu_9$};
    \node[mynode] (5) [right of=9]{$\nu_5$}; 
	\node[mynode] (6) [right of=5] {$\nu_6$};
	\node[mynode] (7) [below of=6] {$\nu_7$};
    \node[mynode] (8) [below of=9] {$\nu_8$};
    \draw[red, dash pattern=on 1mm off 1mm][->](1) -- node[midway, above] {}(2);
	\draw[<-] (1) -- node[midway, left] {}(3);
	\draw[->] (2) -- node[midway, right] {}(4);
	\draw[<-, color=red, dash pattern=on 1mm off 1mm] (3) -- node[midway, below] {}(4);
    \draw[->](5) -- node[midway, above] {}(6);
    \draw[->](6) -- node[midway, above] {}(7);
    \draw[red, dash pattern=on 1mm off 1mm][->](7) -- node[midway, above] {}(5);
    \draw[red, dash pattern=on 1mm off 1mm][->](2) -- node[midway, above] {}(9);
    \draw[->](5) -- node[midway, above] {}(9);
    \draw[red, dash pattern=on 1mm off 1mm][->](8) -- node[midway, above] {}(9);
	\end{tikzpicture}
\caption{A signed digon sign-symmetric digraph containing $3$ leader groups, where the black edges represent cooperative interactions and the dashed red edges represent antagonistic interactions. The first leader group is a SB-rooted SCC containing the leader nodes $\nu_1, \nu_2, \nu_3,$ and $\nu_4$, the second leader group is a SUB-rooted SCC containing the leader nodes $\nu_5, \nu_6,$ and $\nu_7$, and the third leader group is a root (leader) node, $\nu_8$. The node $\nu_9$ is the follower node.} \label{leadergroups}
\end{figure}
We now recall the definitions of the signed incidence matrices of a signed digraph \cite{du2019further,sekercioglu2024control}. 
Consider a signed graph $\mathcal{G}_s$ containing $N$ nodes and $M$ edges. The signed
incidence matrix $E_s \in \mathbb{R}^{N \times M}$ of $\mathcal{G}_s$ is defined as ${[E_{s}]_{ik}} := +1$, if $\varepsilon_k = (\nu_i, \nu_j)$ or if $\nu_i,\nu_j$ are competitive and $\varepsilon_k = (\nu_j, \nu_i$; ${[E_{s}]_{ik}} := -1$, if $\nu_i,\nu_j$ are cooperative and $\varepsilon_k = (\nu_j, \nu_i)$; ${[E_{s}]_{ik}} := 0$, otherwise.
The signed in-incidence matrix $E_{{s \odot}} \in \mathbb{R}^{N \times M}$ of $\mathcal{G}_s$ is defined as ${[E_{s \odot}]_{ik}} := -1$, if $\nu_i,\nu_j$ are cooperative and $\varepsilon_k = (\nu_j, \nu_i)$; ${[E_{s \odot}]_{ik}} := +1$, if $\nu_i,\nu_j$ are competitive and $\varepsilon_k = (\nu_j, \nu_i)$; ${[E_{s \odot}]_{ik}} := 0$, otherwise.
Here, $\varepsilon_k$ is the oriented edge interconnecting nodes $\nu_i$ and $\nu_j$,\ $k \leq M,\ i,j \leq N$. The signed Laplacian $L_{s} \in \mathbb{R}^{N \times N}$ and the signed edge Laplacian $L_{e_{s}} \in \mathbb{R}^{M \times M}$ of 
a signed digraph are given as
$L_{s} = E_{s \odot}E_{s}^\top,\ L_{{e_s}} = E_{s}^\top E_{{s \odot}}.$
The signed Laplacian of a signed digraph is not symmetric, and its eigenvalues all have nonnegative real parts. 

\subsection{Stability of cooperative OMAS}\label{section3b} 
Before delving into the details, we first recall the following crucial definitions and theorem from \cite{xue2022stability} for cooperative OMAS, which serve as the basis for our analysis of the stability of the system. 

Let $\sigma : \mathbb{R}_{\geq 0} \to \mathcal{P}$ be the switching signal associated with topology changes, where $\mathcal{P} := \{1, 2, \dots, s\}$ represents the set of $s$ possible switching modes. 
Each mode of the system is denoted by $\phi \in \mathcal{P}$, with $\phi = \sigma(\tau)$ for $\tau \in [t_l, t_{l+1})$, where $t_l$ and $t_{l+1}$ are consecutive switching instants. A switching instant $t_l$ is defined as the time at which a new edge is created or being removed between agents, a new agent enters or leaves the system, or the sign/direction of an interconnection changes.
Consider the following switched system modeled by
\begin{align}\label{app1}
    \dot{x}_{\sigma(t)}(t) = f_{\sigma(t)}(x_{\sigma(t)}(t)),
\end{align}
where $x_{\sigma(t)} \in \mathbb{R}^{N_{\sigma(t)}}$ is the state. $N_{\sigma(t)}$ indicates that the system's dimension may change with each switch, with $N_{\phi}<+\infty$. For this system, we have the following.
\begin{definition}
The system \eqref{app1} is said to be globally uniformly practically stable, if there exist a class $\mathcal{KL}$ function $\beta$ and a scalar $\epsilon \geq 0$ such that for any initial state $x_{\sigma(t_0)}(t_0)$ and admissible $\phi$,
$\|x_{\sigma(t)}(t)\| \leq \beta(\|x_{\sigma(t_0)}(t_0)\|, t - t_0) + \epsilon, \ \forall t \geq t_0,$
where $\epsilon$ is called the ultimate bound of $x_{\sigma(t)}(t)$ as $t \to +\infty$. 
\end{definition}
Next, we introduce the concept of transition-dependent average dwell time for the switching signal $\sigma(t)$, which ensures that the switching signal meets the required conditions for stability--- See Theorem \ref{thmxue} below and Theorem \ref{prop:result1} in Section \ref{section5}. 

\begin{definition}(\cite[Definition 2]{xue2022stability}) \label{avr_dwell_time_slow} On a given interval $[t_0, t_f)$, with $t_f > t_0 \geq 0$, consider any two consecutive modes $\hat{\phi}, \phi \in \mathcal{P}$, where $\hat{\phi}$ precedes $\phi$. Let $N_{\phi,\hat{\phi}}(t_0, t_f)$ denote the total number of switchings from mode $\hat{\phi}$ to mode $\phi$ within $[t_0, t_f)$, and let $T_{\phi}(t_0, t_f)$ denote the total active duration of mode $\phi$ within $[t_0, t_f)$. The constant $\tau_{\phi,\hat{\phi}} > 0$ satisfying
$    N_{\phi,\hat{\phi}}(t_0, t_f) \leq \hat{N}_{\phi,\hat{\phi}} + \frac{T_{\phi}(t_0, t_f)}{\tau_{\phi,\hat{\phi}}},$
for some given scalar $\hat{N}_{\phi,\hat{\phi}} \geq 0$, is called the slow transition-dependent average dwell time of the switching signal $\sigma(t)$. 
\end{definition}

\begin{thm} (\cite{xue2022stability})\label{thmxue}
    Consider the system \eqref{app1} with the switching signal $\sigma(t)$ on $[t_0,t_f],$ $0\leq t_0 < t_f < + \infty$. Assume that, for any two consecutive modes $\phi$, $\hat \phi \in \mathcal{P}$, where $\hat \phi$ precedes $\phi$, there exist class $\mathcal{K}_{\infty}$ functions $\underline{\kappa}$, $\overline{\kappa}$, constants $\gamma_{\phi},$ $\Omega_{\phi, \hat \phi}>0$, $\Theta \geq 0$, and a non-negative function $V_{\phi}(t, x_{\phi}(t)) : \mathbb{R}_{\geq 0} \times \mathbb{R}^{N_{\phi}} \to \mathbb{R}_{\geq 0}$, such that  $\forall t \in [t_0, t_f]$, 
    \begin{align*} 
        &\underline{\kappa} (\lvert x_{\phi} \rvert) \leq V_{\phi}(t, x_{\phi}(t)) \leq \overline{\kappa} (\lvert x_{\phi} \rvert) \\
        &\dot V_{\phi}(t, x_{\phi}(t)) \leq -\gamma_{\phi} V_{\phi}(t, x_{\phi}(t))\\
        &V_{\phi}(t_k^+, x_{\phi}(t_k^+)) \leq \Omega_{\phi, \hat \phi} V_{\hat \phi}(t_k^-, x_{\hat \phi}(t_k^-)) + \Theta,
    \end{align*}
     where $\gamma_{\phi} > 0$ and $\Omega_{\phi, \hat \phi}>1$. 
     Moreover, asssume that $\sigma(t)$ satisfies
        $\tau_{\phi,\hat{\phi}} \geq \frac{\ln(\Omega_{\phi,\hat{\phi}})}{\gamma_{\phi}},$
    for any $\hat \phi \in \mathcal{P}$ 
    where $\tau_{\phi,\hat{\phi}}$ is defined in Definition \ref{avr_dwell_time_slow}. 
    Then, \eqref{app1} is globally uniformly practically stable. 
\end{thm} 

\section{Model and Problem Formulation}\label{section2}
We consider a network of agents interconnected over a signed digraph with cooperative and antagonistic interactions. The considered digraph is dynamic, meaning that agents can join or leave, edges may be added or deleted, and the sign or direction of the interactions can change over time. These changes in the MAS can be modeled using a switched system representation as in Section \ref{section3b}. For an illustration of different modes in a signed OMAS, see Figure \ref{graph}.

At each mode $\phi \in \mathcal{P}$, consider the MAS composed of $N_{\phi}$ dynamical systems modeled by
\begin{align}\label{FO}
	\dot x_i = u_i,\quad i \in \{1, 2, \dots, N_{\phi} \},
\end{align}
where $x_i \in \mathbb{R}$ is the state of the $i$th agent, and $u_i \in \mathbb{R}$ is the control input. For notational simplicity and without loss of generality, we assume that $x_i \in \mathbb{R}$, but all contents of this paper apply to systems of higher dimension $x_i \in \mathbb{R}^n, n \geq 1$, using a Kronecker product. 

The agents interact on a dynamic signed digraph $\mathcal{G}_{s _{\phi}}$ that contains $N_{\phi}$ nodes and $M_{\phi}$ edges. The system \eqref{FO} is interconnected with the control law
\begin{align}\label{CL}
	u_i = -k_{1} \sum_{j=1}^{N_{\phi}} \lvert a_{\phi_{ij}} \rvert \left[ x_i - \mbox{sign}(a_{\phi_{ij}} )x_j \right], 
\end{align}
where $k_{1}>0$, and $A_{\phi} = [a_{\phi_{ij}}]$ is the adjacency matrix at mode $\phi \in \mathcal{P}$, with $a_{\phi_{ij}} \in \{ 0, \pm 1 \}$ representing the adjacency weight between nodes $\nu_j$ and $\nu_i$. It is well known that under the distributed control law \eqref{CL}, agents interconnected over a SB signed digraph achieve bipartite consensus if and only if the signed graph contains a directed spanning tree, while those over a SUB digraph reach trivial consensus provided that the graph contains a directed spanning tree and has no root node \cite{hu2014emergent,altafini2012_6329411}. Thus, we pose the following assumption on the connectivity of the initial signed graph.
\begin{ass}\label{ass1}
    The initial signed digraph contains a directed spanning tree.
\end{ass}

We consider a dynamic signed digraph with node additions and removals to model systems where agents progressively join or leave the system over time. Graph connectivity is crucial for sustaining the desired collective behavior of the network: for this reason, in our formulation, we assume that at each mode the signed digraph is weakly connected, which can be ensured by allowing only newly added agents to be removed\footnote{Our results can be extended to include disconnected signed digraphs for each mode. However, this would imply an extension of the derivations provided in Section~\ref{section4}. Therefore, we omit this extension for lack of space and to preserve the discussion on synchronization.}.

We are interested in guaranteeing the stability of the OMAS under all and possibly infinite switching modes $\mathcal{P}$. 
The possible achievable control objectives for the system \eqref{FO} interconnected with the control law \eqref{CL} over a signed digraph depend on the \textit{structural balance} property and the graph topology. In particular, the presence of rooted SCCs and root nodes plays a crucial role in determining the agents' convergence behavior. To formulate the achievable control objective for the OMAS, it is important to address all achievable control objectives for each mode. On the one hand, a rooted SCC influences the convergence of the rest of the network, depending on whether it is SB or SUB. If the digraph contains a single root node, the system exhibits a leader-following bipartite consensus, where all agents converge either to the leader's state or its opposite state if the digraph is SB. In such cases, the root node or the rooted SCC acts as a \textit{leader group}, influencing the remaining agents, referred to as \textit{followers}. More formally, we define a leader group as either a single root node or a rooted SCC consisting of multiple nodes. On the other hand, if the digraph contains multiple root nodes or rooted SCCs, each of these leader groups dictates the behavior of the remaining agents. Instead of converging to a single or bipartite equilibrium, the agents settle within a region defined by the states of the leader groups. In this particular case, the digraph no longer contains a directed spanning tree, as independent leader groups are not mutually reachable. Therefore, for signed digraphs with multiple leaders groups, we pose the following connectivity assumption. 
\begin{ass}\label{ass2} The signed digraph contains $m$ leader groups, where $l_1$ is the number of root nodes, $l_{2_{SB}}$ is the number of SB-rooted SCCs, and $l_{2_{SUB}}$ is the number of SUB-rooted SCCs with $m=l_1 + l_2 > 1$ and $l_2 = l_{2_{SB}} +l_{2_{SUB}}$; $|\mathcal L|=k <N$ and $\mathcal F$ are the sets containing the indices corresponding to the nodes in the leader groups and followers respectively.
\begin{enumerate}
	\item The $k$ leader nodes are organized into $m$ leader groups of $p_i$ nodes, each forming a strongly connected subgraph (or $p_i=1$ if it is a single root node), where $1 < m \leq k,\ i \leq m$, and $\sum_{i=1}^m p_i = k$.
	\item There exists at least one path from $\nu_i$ to $\nu_j$ for all $j \in \mathcal{F}$ and $i \in \mathcal{L}$.
\end{enumerate}
\end{ass}

The achievable control objective for \eqref{FO} in closed loop with the distributed control law \eqref{CL} and interconnected over a SB digraph:
\begin{itemize}
    \item containing a directed spanning tree (Assumption \ref{ass1}) is to ensure agents achieve \textit{bipartite consensus},  where agents converge to the same value in modulus but not in sign, that is,
$	\lim_{t \to \infty} \left[ x_{i}(t) - \mbox{sign}(a_{\phi_{ij}})x_{j}(t) \right] = 0, \ \forall i,j \leq N_{\phi}.$
    \item containing $m$ leader groups, under Assumption \ref{ass2}, is to ensure agents achieve \textit{bipartite containment}, that is, $ \lim_{t \to \infty} \big[\, \lvert x_{j}(t) \rvert - \max_{i \in \mathcal{L}}\lvert x_{i}(t) \rvert \,\big] \leq 0,$ for each $j \in \mathcal{F}$.
\end{itemize} 

The achievable control objective for \eqref{FO} in closed loop with the distributed control law \eqref{CL} interconnected over a SUB digraph:
\begin{itemize}
    \item containing a directed spanning tree (Assumption \ref{ass1}) and a SUB-rooted SCC, without a root (leader) node, is to ensure agents achieve \textit{trivial consensus}, where all agents converge to zero, that is,
	   $\lim_{t \to \infty} x_{i}(t) = 0, \quad \forall i \leq N_{\phi}.$
    \item containing a directed spanning tree (Assumption \ref{ass1}) and SB-rooted SCC or a root node, is to ensure agents achieve \textit{interval bipartite consensus}, that is,
	  $ \lim_{t \to \infty} x_{i}(t) = x^* \in [-\theta, \theta], \ \forall i \leq N_{\phi},$
    where $\theta > 0$.
    \item containing $m$ leader groups, under Assumption \ref{ass2}, is to ensure agents achieve \textit{bipartite containment}, that is, 
	  $ \lim_{t \to \infty} \big[\, \lvert x_{j}(t) \rvert - \max_{i \in \mathcal{L}}\lvert x_{i}(t) \rvert \,\big] \leq 0,\ j \in \mathcal{F}.$
\end{itemize}

In this article, we show that, under Assumptions~\ref{standing_ass} and \ref{ass1}, the signed OMAS is uniformly practically stable. 
In the case agents achieve trivial consensus or bipartite consensus, their edges states, defined as
\begin{equation}\label{def_e}
    e_{\phi_k} = x_{i} - \mbox{sign}(a_{\phi_{ij}})x_{j},\quad \varepsilon_k =(\nu_j,\nu_i) \in \mathcal{E}_{\phi}
\end{equation}
where $k \leq M_{\phi}$ denotes the index of the interconnection between the $j$th and $i$th agents, converge to zero ---See \cite{CDC2025OMAS}. However, as the networks considered here, in addition to signs on the edges, contain a priori, rooted SCCs or multiple root nodes, the resulting Laplacians can also have multiple zero eigenvalues in some cases. This also results, in general, in multiple convergence points for agents and their edge states, which means that edge states do not always converge to zero. Rather than handling multiple control objectives individually, we express them through a single equivalent objective. Then, following the framework laid in \cite{DYNCON-TAC16} and
extending it to signed networks with associated Laplacians containing multiple zero eigenvalues and to the signed edge-based formulation, we define the weighted average system for the edge states.

Let $\xi_{\phi}$ be the number of zero eigenvalues of $L_{e_{s_\phi}}$ in mode $\phi \in \mathcal{P}$. Then, we define the weighted edge average state as
\begin{equation}\label{e_m}
	e_{m_{\phi}} := \sum_{i=1}^{\xi_{\phi}} v_{r_{\phi_i}}v_{l_{\phi_i}}^{\top} e_\phi,
\end{equation}
where $v_{r_{\phi_i}}$ and $v_{l_{\phi_i}}$ are the right and left eigenvectors associated with the zero eigenvalues of the edge Laplacian, $e_\phi := [e_{\phi_1} \ e_{\phi_2} \ \cdots \ e_{\phi_{M_\phi}}]^\top$, and $e_{\phi_k}$ is defined in \eqref{def_e}. We define the synchronization errors as
\begin{align}\label{def_e_bar}
	\bar e_{{\phi}} &= e_\phi - e_{m_{\phi}} = [I - \sum_{i=1}^{\xi_{\phi}} v_{r_{\phi_i}}v_{l_{\phi_i}}^{\top}]e_\phi,
\end{align}
where $\bar{e}_{\phi} := [\bar e_1 \ \bar e_2 \ \dots \ \bar e_{M_\phi}]^\top$. Then, the control objective is equivalent to making the synchronization errors converge to zero, that is, 
\begin{align}\label{obj}
	\lim_{t \to \infty} \bar{e}_{\phi_k}(t) = 0, \quad \forall k \leq M_{\phi}.
\end{align}
\begin{remark}
In the case where the nullspaces of $L_{e_s}^\top$ and $E_s$ are equal to each other, $e_{m_{\phi}} = 0$ and the synchronization errors are equal to the edge states.
\end{remark}

\section{Lyapunov Equation for Directed Signed Edge Laplacians}\label{section4}
To establish synchronization of OMAS over signed digraphs, we prove the asymptotic practical stability of the set $\{ \bar e_\phi = 0\}$. In particular, we show how to construct strict Lyapunov functions, in the space of $\bar e_\phi$, referring to functions expressed in terms of $\bar e_\phi$, based on the following statements, which are original contributions of this paper and extend \cite[Proposition~1]{ECC20_EP} to the case of directed and signed Laplacians containing multiple zero eigenvalues and \cite[Proposition 1]{csekerciouglu2024distributed} to the case of signed edge Laplacians corresponding to directed graphs with multiple zero eigenvalues (Theorems~\ref{thm:lyap_eq_Le_directed} and \ref{proposition:lyap_eq}).

For a directed spanning tree graph, we have the following.

\begin{thm} \label{lemma:lyap_eq_spanningtree} Let $\mathcal{G}_{s}$ be a signed digraph containing $N$ agents interconnected by $M$ edges, and let $L_{e_{s}}$ be the associated edge Laplacian. If the graph $\mathcal{G}_{s}$ is a directed spanning tree, then for any $Q \in \mathbb{R}^{(N-1) \times (N-1)}, Q=Q^{\top}> 0$, there exists a matrix $P \in \mathbb{R}^{(N-1) \times (N-1)}$, $P=P^{\top}>0$ such that 
\begin{equation}\label{Lyap-eq}
    PL_{e_{s}} + L_{e_{s}}^{\top}P = Q.
\end{equation}
\end{thm} 

\begin{proof} By assumption, the graph $\mathcal{G}_{s}$ is a directed spanning tree, and is therefore SB. Then, it consists of $N-1$ edges, and $L_{e_{s}} \in \mathbb{R}^{(N-1)\times (N-1)}$. Since the nonzero eigenvalues of $L_{e_{s}}$ and $L_{{s}}$ coincide, the rank of $L_{e_{s}}$ is $N-1$ and $L_{e_{s}}$ has eigenvalues only with positive real parts, so $-L_{e_{s}}$ is Hurwitz. Then, given any $Q = Q^\top >0,\ Q \in \mathbb{R}^{(N-1) \times (N-1)}$, there exists a symmetric positive definite matrix $P \in \mathbb{R}^{(N-1) \times (N-1)}$ such that $PL_{e_{s}} + L_{e_{s}}^{\top}P = Q$ holds.\end{proof}

For a signed digraph containing a directed spanning tree, we have the following.

\begin{thm} \label{thm:lyap_eq_Le_directed} Let $\mathcal{G}_s$ be a signed digraph of $N$ agents interconnected by $M$ edges. Let $L_{e_{s}}$ be the associated directed edge Laplacian, which contains $\xi$ zero eigenvalues. Then, the following are equivalent:
	\renewcommand{\theenumi}{(\roman{enumi}}
	\begin{enumerate}
		\item $\mathcal{G}_s$ contains a directed spanning tree;
        \item For any $Q \in \mathbb{R}^{M\times M}, Q=Q^{\top}> 0$, and for any $\{ \alpha_{1}, \alpha_{2}, \dots, \alpha_{\xi}\}$ with $\alpha_{i} >0$, there exists a matrix $P(\alpha_{i}) \in \mathbb{R}^{M \times M}$, $P=P^{\top}>0$ such that 
\begin{equation}\label{Lyap_eq_dir}
    PL_{e_{s}} + L_{e_{s}}^{\top}P = Q - \sum_{i=1}^{\xi} \alpha_{i} \left(Pv_{r_{i}}v_{l_{i}}^{\top} + v_{l_{i}}v_{r_{i}}^{\top}P\right),
\end{equation}
where $v_{r_{i}}, v_{l_{i}} \in \mathbb{R}^M$ are, respectively, the normalized right and left eigenvectors of $L_{e_{s}}$ associated with the $i$th 0 eigenvalue. 
\end{enumerate}
$\xi$ satisfies $\xi = M-N+1$ if the signed digraph is SB, and also if it is SUB with a root node or SB-rooted SCC; otherwise, $\xi = M-N$.

\end{thm} 
\begin{remark}
In the case the considered digraph contains $N$ edges and is SUB without a root node or without a SB-rooted SCC, \eqref{Lyap_eq_dir} holds with $\xi=0$, so is equivalent to \eqref{Lyap-eq}. 
\end{remark}

\begin{proof}
(i) $\Rightarrow$ (ii): By assumption, the graph \(\mathcal{G}_{s}\) contains a directed spanning tree. Then, from Lemmata 4 and 5 from \cite{Pelin_edge}, it follows that the eigenvalues of \(L_s\) depend on the structure of \(\mathcal{G}_{s}\). Moreover, from Item (iii) of \cite[Lemma 4]{Pelin_edge}, the nonzero eigenvalues of $L_{s}$ and $L_{e_{s}}$ coincide. If \(\mathcal{G}_{s}\) is SB or SUB with either a root node or an SB-rooted SCC, \(L_s\) has a unique zero eigenvalue, with the remaining \(N-1\) eigenvalues having positive real parts, i.e., $0 = \lambda_1 < \Re e (\lambda_2) \leq \dots \leq \Re e (\lambda_{N})$. If \(\mathcal{G}_{s}\) instead contains a SUB-rooted SCC, then all \(N\) eigenvalues have positive real parts, i.e., $0 < \Re e (\lambda_1) < \Re e (\lambda_2) \leq \dots \leq \Re e (\lambda_{N})$. $L_{e_{s}}$ has $\xi$ zero eigenvalues: $0 = \lambda_1 = \dots = \lambda_{\xi} < \Re e (\lambda_{\xi+1}) \leq \dots \leq \Re e (\lambda_{M})$, where $\xi = M - N + 1$ if the graph is SB or either contains either a root node or an SB-rooted SCC, and $\xi = M - N$, otherwise. Then, from Eq. (4) in \cite{Pelin_edge}, we can write the Jordan decomposition\footnote{This is the general Jordan decomposition of $L_{e_{s}}$. Based on the graph structure, some of the elements may not appear.} of $L_{e_{s}}$ as $L_{e_{s}} = U \Lambda U^{-1} = \sum_{i=1}^{\xi} \lambda_{i}(L_{e_{s}})v_{r_{i}}v_{l_{i}}^{\top} + v_{r_{1}}v_{l_{2}}^{\top} +U_{1} \Lambda_{1} U_{1}^{-1}$ where $ \Lambda_{1} \in \mathbb{C}^{(M - \xi) \times (M - \xi)}$ contains the Jordan blocks of $L_{e_{s}}$ corresponding to the eigenvalues with positive real parts. The matrix $U$ is given by $U = \left[ v_{r_{1}}, \cdots, v_{r_{{\xi}}}, U_{1} \right] \in \mathbb{C}^{M \times M}$, and its inverse is $U^{-1} = \left[ v_{l_{1}}^\top, \cdots, v_{l_{{\xi}}}^\top, U_{1}^\dagger \right]^\top \in \mathbb{C}^{M \times M}$. Here, \( v_{r_{i}} \) and \( v_{l_{i}} \) represent the right and left eigenvectors of \(L_{e_{s}}\), respectively. The matrices \(U_{1}\) and \(U_{1}^\dagger\) contain the remaining \(M-\xi\) columns of \(U\) and \(U^{-1}\), respectively. From \cite{Horn}, the inner product between eigenvectors corresponding to distinct eigenvalues is zero. Since $v_{r_{1}}$ and $v_{l_{2}}$ are associated with different eigenvalues, we have $v_{l_{2}}^{\top}v_{r_{1}}=0$. Consequently, the matrix $v_{r_{1}}v_{l_{2}}^{\top}$, whenever present, has only zero eigenvalues, and has no effect. Then, for any $\{ \alpha_{1}, \alpha_{2}, \dots, \alpha_{{\xi}}\}$ with $\alpha_{i} >0$, define $R = L_{e_{s}} + \sum_{i=1}^{\xi} \alpha_{i}v_{r_{i}}v_{l_{i}}^{\top}$. From this decomposition and the fact that $\Lambda_{1}$ contains only eigenvalues with positive real parts, it is clear that $ \Re e \{\lambda_{j}(R) \} > 0$ for all $j \leq M$. That is, $-R$ is Hurwitz; therefore, for any ${Q = Q^{\top}>0}$ and $\alpha_{i} >0, i \leq \xi$, there exists $P = P^{\top}>0$ such that $-PR-R^{\top}P = -Q$. This is equivalent to
$		-P [ L_{e_{s}} + \sum_{i=1}^{\xi} \alpha_{i}v_{r_{i}}v_{l_{i}}^{\top} ] - [L_{e_{s}} + \sum_{i=1}^{\xi} \alpha_{i}v_{r_{i}}v_{l_{i}}^{\top} ]^{\top}P = -Q
	\Leftrightarrow PL_{e_{s}} + L_{e_{s}}^{\top}P = Q - \sum_{i=1}^\xi \alpha_i (Pv_{r_i}v_{l_i}^{\top} + v_{l_i}v_{r_i}^{\top}P).$

	(ii) $\Rightarrow$ (i): Assume that the signed edge Laplacian has $\xi+1$ zero eigenvalues and the rest of its $M-\xi-1$ eigenvalues have positive real parts. In view of Lemmata 4 and 5 from \cite{Pelin_edge}, neither the property that $L_{e_{s}}$ has $\xi$ eigenvalues with positive real parts ---with $\xi = M - N + 1$ if the graph is SB or contains a root node or a SB-rooted SCC, and $\xi = M - N$, otherwise--- nor the fact that the graph contains a directed spanning tree hold. Now, the Jordan decomposition of $L_{e_{s}}$ has the form $L_{e_{s}} = U \Lambda U^{-1} = \sum_{i=1}^{\xi+1} \lambda_{i}(L_{e_{s}})v_{r_{i}}v_{l_{i}}^{\top} + v_{r_{1}}v_{l_{2}}^{\top} + U_{1} \Lambda_{1} U_{1}^{-1}$ with $U = \left[ v_{r_{1}}, \dots, v_{r_{{\xi +1}}}, U_{1} \right] \in \mathbb{R}^{M \times M}$ and $U^{-1} = \left[ v_{l_{1}}, \dots, v_{l_{{\xi +1}}}, U_{1}^{\dagger} \right]^\top \in \mathbb{R}^{M \times M}$. Let us consider $R(\alpha_{i}) =L_{e_{s}} + \sum_{i=1}^{\xi} \alpha_{i}v_{r_{i}}v_{l_{i}}^{\top}$ which admits the Jordan decomposition $R(\alpha_{i}):= U \Lambda_{R} U^\top$, where 
	$$
	\Lambda_{R}(\alpha_{i}) := \begin{bmatrix} \alpha_{1} &  &  & & \\[-5pt]  &\hspace{-1.3ex} \ddots & & & \\[-5pt]  &  &\hspace{-1ex} \alpha_{{\xi}} & & \\[1pt]   &  & &\hspace{-1em} 0 & \\[-1pt]  & &  & &\hspace{-1ex}  \Lambda_{1} \end{bmatrix}.
	$$
    Since we assume that $L_{e_s}$ has $\xi+1$ zero eigenvalues, $R$ has one zero eigenvalue, and this cannot be positive definite. Then, there exists a matrix $Q = Q^{\top}>0$ for which no symmetric matrix $P=P^{\top}$ satisfies $-PR-R^{\top}P = -Q$. This contradicts statement (ii), completing the proof. \end{proof}

For a signed digraph containing multiple leaders, we have the following.
\begin{thm} \label{proposition:lyap_eq} Let $\mathcal{G}_s$ be a signed digraph of $N$ agents interconnected by $M$ edges, containing $m$ leader groups. Let $l_1$ be the number of root nodes, $l_{2_{SB}}$ be the number of SB-rooted SCCs, and $l_{2_{SUB}}$ be the number of SUB-rooted SCCs, where $m=l_1 + l_2$ and $l_2 = l_{2_{SB}} +l_{2_{SUB}}$. Let $L_{e_{s}}$ be the associated directed edge Laplacian. Then the following are equivalent:
\renewcommand{\theenumi}{(\roman{enumi}}
\begin{enumerate}
	\item the graph has $m$ leader groups, and given each follower $\nu_j, \forall j \in \mathcal{F}$, there exists at least one leader $\nu_i, i \in \mathcal{L}$, such that there exists at least one path from $\nu_i$ to $\nu_j$,
	\item for any $Q \in \mathbb{R}^{M \times M}, Q=Q^{\top}> 0$ and for any $\{ \alpha_1, \alpha_2, \dots, \alpha_{\xi}\}$ with $\alpha_i >0$, there exists a matrix $P(\alpha_i) \in \mathbb{R}^{M \times M}$, $P=P^{\top}>0$ such that \eqref{Lyap_eq_dir} holds,
where $v_{r_i}, v_{l_i} \in \mathbb{R}^M$ are the right and left eigenvectors of $L_{e_{s}}$ associated with the $i$th 0 eigenvalue.
\end{enumerate}
Moreover, $\xi$ satisfies $\xi = M-N+l_1+l_{2_{SB}}$ whether the signed digraph is SB or SUB.
\end{thm}
\begin{proof}
(i) $\Rightarrow$ (ii): By assumption, the graph \(\mathcal{G}_{s}\) contains multiple leader groups, as defined by Assumption \ref{ass2}. Then, from Lemma 7 in \cite{Pelin_edge}, it follows that $L_s$ has $l_1+l_{2_SB}$ zero eigenvalues, with the remaining eigenvalues having positive real parts, i.e., $0 = \lambda_1 = \dots = \lambda_{l_1+l_{2_SB}} < \Re e (\lambda_{l_1+l_{2_SB}+1}) \leq \dots \leq \Re e (\lambda_{N})$. Since the nonzero eigenvalues of $L_{s}$ and $L_{e_{s}}$ coincide, $L_{e_{s}}$ has $\xi$ zero eigenvalues: $0 = \lambda_1 = \dots = \lambda_{\xi} < \Re e (\lambda_{\xi+1}) \leq \dots \leq \Re e (\lambda_{M})$, where $\xi = M-N+l_1+l_{2_{SB}}$.  Let $\gamma$ and $\xi$ be the geometric and algebraic multiplicities of the zero eigenvalues. Then, from Eq. (4) in \cite{Pelin_edge}, we can write the Jordan decomposition of $L_{e_{s}}$ as $L_{e_{s}} = U \Lambda U^{-1} = \sum_{i=1}^{\xi} \lambda_{i}(L_{e_{s}})v_{r_{i}}v_{l_{i}}^{\top} + \sum_{i=1}^{\xi-\gamma} v_{r_{2i-1}}v_{l_{2i}}^{\top} +U_{1} \Lambda_{1} U_{1}^{-1}$. On the other hand, from \cite{Horn} the inner product between eigenvectors corresponding to distinct eigenvalues is zero. Since $v_{r_{2i-1}}$ and $v_{l_{2i}}$ are associated with different eigenvalues, we have $v_{r_{2i-1}}^{\top}v_{l_{2i}}=0$ for $i \in \{1,2,\dots, 2(\xi-\gamma)\}$. Consequently, the matrix $v_{r_{2i-1}}v_{l_{2i}}^{\top}$ has only zero eigenvalues. The proof follows by the steps of the proof of Theorem \ref{thm:lyap_eq_Le_directed}.
\end{proof}

\section{Main Results}\label{section5}
This section presents our main results. We consider the synchronization problem of OMAS modeled by \eqref{FO} over dynamic directed signed graphs, and establish the practical asymptotic stability of the synchronization errors \eqref{def_e_bar}. First, we derive the synchronization errors using the edge-Laplacian notation introduced in Section \ref{section3}, and then reformulate the control problem accordingly. 

We now proceed to derive the synchronization errors using the edge-Laplacian notation from Section \ref{section3}. At each mode $\phi \in \mathcal{P}$, using the definition of the incidence matrix, we may express the edge states in \eqref{def_e} in vector form
\begin{align}\label{err_vect}
	e_{\phi} = E_{s_{\phi}}^\top x,
\end{align}
where $E_{s_{\phi}}$ is the incidence matrix corresponding to the signed digraph at mode $\phi$. From the definition of the directed Laplacian, we write the control law in \eqref{CL} in vector form as
\begin{align}\label{CL_vect_edge}
	u_{\phi} = -k_1  E_{{s \odot}_{\phi}} {e}_{{\phi}}.
\end{align}
Differentiating the edge states \eqref{err_vect} yields
\begin{align}\label{edgestates}
	\dot{e}_{{\phi}} = -k_1 E_{s_{\phi}}^\top E_{{s \odot}_{\phi}} e_{\phi} =  -k_1 L_{e_{s_{\phi}}} e_{\phi}.
\end{align}
Similarly, differentiating the synchronization errors \eqref{def_e_bar},
$	\dot{\bar{e}}_{{\phi}} = [I -\sum_{i=1}^{\xi_{\phi}} v_{r_{\phi_i}}v_{l_{\phi_i}}^{\top} ] \dot{e}_{{\phi}} = -k_1 [I -\sum_{i=1}^{\xi_{\phi}} v_{r_{\phi_i}}v_{l_{\phi_i}}^{\top} ] L_{e_{s_{\phi}}}{e}_{{\phi}},$
we obtain
\begin{equation}\label{27}
	\dot{\bar{e}}_{{\phi}}= -k_1 L_{e_{s_{\phi}}} \bar{e}_{{\phi}}.
\end{equation}
\begin{remark}
Observe that, to obtain \eqref{27}, we distinguish two cases based on \cite[Table 1, Eq. (5)]{Pelin_edge}: the first case where $\xi_{\phi} = \gamma_{\phi}$, and the second where $\xi_{\phi}>\gamma_{\phi}$. Here, $\gamma_{\phi}$ and $\xi_{\phi}$ denote the geometric and algebraic multiplicities of the zero eigenvalue in mode $\phi \in \mathcal{P}$, respectively.
\begin{itemize}
    \item Case 1 ( $\xi_{\phi} = \gamma_{\phi}$ ): In this case, all left eigenvectors satisfy $v_{l_{\phi_i}}^{\top}L_{e_{s_\phi}}=0.$ Consequently, we have $\dot{\bar{e}}_{{\phi}} = -k_1L_{e_{s_{\phi}}} e_{{\phi}}.$ Additionally, since $L_{e_{s_\phi}}v_{r_{\phi_i}}=0$, it also holds that $\dot{\bar{e}}_{{\phi}} = -k_1L_{e_{s_{\phi}}}[I -\sum_{i=1}^{\xi_{\phi}} v_{r_{\phi_i}}v_{l_{\phi_i}}^{\top} ]{e}_{{\phi}}.$ Finally, substituting \eqref{def_e_bar} in the latter, we obtain \eqref{27}.
    \item Case 2 ( $\xi_{\phi}>\gamma_{\phi}$ ): For $i \in \{1, 2, \dots, \xi_{\phi}-\gamma_{\phi} \}$, we have the relations $v_{l_{\phi_{2i-1}}}^{\top}L_{e_{s_\phi}}=v_{l_{\phi_{2i}}}^{\top},$ $v_{l_{\phi_{2i}}}^{\top}E_{s_{\phi}}^\top=0,$ and $L_{e_{s_\phi}}v_{r_{\phi_{2i}}}=v_{r_{\phi_{2i-1}}}$, $E_{s_{\phi}}v_{r_{\phi_{2i-1}}}=0$. The remaining $2\gamma_\phi - \xi_{\phi}$ right and left eigenvectors associated with zero eigenvalues satisfy $L_{e_{s_\phi}}v_{r_{\phi_i}}=0$ and $v_{l_{\phi_i}}^{\top}L_{e_{s_\phi}}=0.$ Thus, for this case as well, $\dot{\bar{e}}_{{\phi}} = -k_1L_{e_{s_{\phi}}}[I -\sum_{i=1}^{\xi_{\phi}} v_{r_{\phi_i}}v_{l_{\phi_i}}^{\top} ]{e}_{{\phi}}$ holds, and substituting \eqref{def_e_bar} in the latter, Eq. \eqref{27} follows.
\end{itemize}
\end{remark}
 The closed-loop system for the synchronization error dynamics is given as
\begin{subequations}\label{dyn_edge1}
\begin{align}
  &\dot{\bar e}_{{\phi}}(t) = -k_1 L_{e_{s_{\phi}}} \bar e_{{\phi}}(t),\quad t \in [t_l, t_{l+1}),\label{dyn_edge}\\
  &\bar e_{{\phi}}(t_l^+) = \Xi_{\phi, \hat \phi} \bar e_{{\hat \phi}}(t_l^-) + \Phi_l,\quad t = t_l. \label{edge_switch}
\end{align}
\end{subequations}
The equation in \eqref{edge_switch} formulates the edge state transition of the switched system at each switching instant \( t_l \), as described in \cite{xue2022stability}, where $\phi = \sigma(t) \in \mathcal{P}$ represents the active mode after switching, with $t \in [t_l, t_{l+1})$, while $\hat \phi = \sigma(\hat t) \in \mathcal{P}$ denotes the previous mode, where $\hat t \in [t_{l-1}, t_l)$. The matrix $\Xi_{\phi, \hat \phi} \in \mathbb{B}^{M_{\phi} \times M_{\hat \phi}}$ is a matrix with elements in $\{ 0, 1\}$ that determines how the dimension of the error state \( \bar e_{{\hat{\phi}}}(t_l^-) \) transitions between modes. Here, \( \bar e_{{\hat{\phi}}}(t_l^-) \) represents the error states just before the switching instant, while \( \bar e_{{\phi}}(t_l^+) \) represents the error states just after the switching instant. $\Phi_l \in \mathbb{R}^{M_{\phi}}$ is a real-valued and bounded vector that captures instantaneous changes in the edge state $e_{{\phi}}$ at the switching moment $t_l$. These changes may arise from the addition or removal of a node, which leads to the creation or deletion of an edge, or from changes in the sign or direction of interactions, which can cause switches between cooperation and antagonism, as well as shifts in leader–follower roles--see Figure \ref{graph}.

We are now ready to present our main result on practical asymptotic stability of the origin of the error dynamics \eqref{dyn_edge1}. The control objective for \eqref{dyn_edge1} is to ensure that the origin is uniformly practically stable. 
\begin{thm}\label{prop:result1}
Consider the signed OMAS \eqref{FO}, under Assumptions \ref{standing_ass} and \ref{ass1}, in closed loop with the switching control law \eqref{CL_vect_edge}. Let $\phi, \hat \phi \in \mathcal{P}$ be any two consecutive modes, where $\hat \phi$ precedes $\phi$. Let $P_{\phi}$ be generated by \eqref{Lyap-eq} if the signed digraph is a directed spanning tree, and by \eqref{Lyap_eq_dir} otherwise, where $\lambda_{\max}(P_{\phi})$ and $\lambda_{\min}(P_{\phi})$ the largest and the smallest eigenvalue of matrix $P_{\phi}$. 

Then, if the switching signal $\sigma$ admits an average dwell time satisfying
	\begin{equation}\label{cond}
		\tau_{\phi,\hat{\phi}} \geq \frac{\ln(\Omega_{\phi,\hat{\phi}})}{\gamma_{\phi}},
	\end{equation} 
	where $\Omega_{\phi,\hat{\phi}}=1 + 2 \frac{\lambda_{\max}(P_{\phi})}{\lambda_{\min}(P_{\hat \phi})}$ and $\gamma_{\phi}= k_1 \frac{1}{\lambda_{\max}(P_{\phi})}$ are positive constants, the origin of the closed-loop system \eqref{dyn_edge1} is uniformly practically stable for all initial conditions.
\end{thm}

\begin{proof} For each mode $\phi$, let $Q_{\phi} = Q_{\phi}^\top > 0$ and $\alpha_{\phi}$ be arbitrarily fixed. By Assumption \ref{ass1} and Theorem \ref{thm:lyap_eq_Le_directed}, there exists a matrix $P_{\phi} = P_{\phi}^\top > 0$ such that \eqref{Lyap_eq_dir} holds, or if the considered graph is a spanning tree, such that \eqref{Lyap-eq} holds. In the case the digraph contains multiple leader groups, by Theorem \ref{proposition:lyap_eq}, there exists a $P_{\phi} = P_{\phi}^\top > 0$ such that \eqref{Lyap_eq_dir} holds. Consider the Lyapunov function candidate
	\begin{align}\label{V}
		V_{\phi}(\bar e_{{\phi}}) := \frac{1}{2} \bar e_{{\phi}}^\top P_{\phi} \bar e_{{\phi}}.
	\end{align}
	For all $\tau \in [t_l, t_{l+1})$, the derivative of \eqref{V} along the trajectories of \eqref{dyn_edge} yields
	$\dot V_{\phi}(\bar e_{{\phi}}(\tau)) = - k_1 \bar e_{{\phi}}(\tau)^\top P_{\phi} L_{e_{s_{\phi}}} \bar e_{{\phi}}(\tau).$
	If the underlying signed digraph is a directed spanning tree, using \eqref{Lyap-eq} with $Q_{\phi} = I_{M_{\phi} \times M_{\phi}}$, we obtain
		$\dot V_{\phi}(\bar e_{{\phi}}(\tau)) = - \frac{1}{2} k_1 \bar e_{{\phi}}(\tau)^\top \bar e_{{\phi}}(\tau). $
	If the underlying signed digraph contains a directed spanning tree or multiple leader groups, using \eqref{Lyap_eq_dir} with $Q_{\phi} = I_{M_{\phi} \times M_{\phi}}$, we obtain
$		\dot V_{\phi}(\bar e_{{\phi}}(\tau)) = - \frac{1}{2} k_1 \bar e_{{\phi}}(\tau)^\top \bar e_{{\phi}}(\tau) 
		 + \frac{1}{2} \bar e_{{\phi}}(\tau)^\top \left( \sum_{i=1}^{\xi_{\phi}} \alpha_{\phi_i} (P_{\phi}v_{r_{\phi_i}}v_{l_{\phi_i}}^{\top} + v_{l_{\phi_i}}v_{r_{\phi_i}}^{\top}P_{\phi}) \right) \bar e_{{\phi}}(\tau).$ 
         From the definition of the synchronization errors \eqref{def_e_bar} and the fact that $v_{l_{\phi_i}}^{\top}v_{r_{\phi_i}}=1$, we have $		\sum_{i=1}^{\xi_{\phi}} \alpha_{\phi_i} P_{\phi}v_{r_{\phi_i}}v_{l_{\phi_i}}^{\top} [I - \sum_{i=1}^{\xi_{\phi}} v_{r_{\phi_i}}v_{l_{\phi_i}}^{\top}]e_{{\phi}}
        = \sum_{i=1}^{\xi_{\phi}} \alpha_{\phi_i} P_{\phi}v_{r_{\phi_i}}v_{l_{\phi_i}}^{\top}e_{{\phi}} - \sum_{i=1}^{\xi_{\phi}} \alpha_{\phi_i} P_{\phi}v_{r_{\phi_i}}v_{l_{\phi_i}}^{\top} e_{{\phi}} = 0$
    and $	\bar e_{{\phi}}^\top [I - \sum_{i=1}^{\xi_{\phi}} v_{r_{\phi_i}}v_{l_{\phi_i}}^{\top}]^\top \sum_{i=1}^{\xi_{\phi}}\alpha_{\phi_i} v_{l_{\phi_i}}v_{r_{\phi_i}}^{\top}P_{\phi}
        = \sum_{i=1}^{\xi_{\phi}} e_{{\phi}}^\top  \alpha_{\phi_i} v_{l_{\phi_i}}v_{r_{\phi_i}}^{\top}P_{\phi} - \sum_{i=1}^{\xi_{\phi}} e_{{\phi}}^\top  \alpha_{\phi_i} v_{l_{\phi_i}}v_{r_{\phi_i}}^{\top}P_{\phi} = 0.$
    Thus,
	$	\dot V_{\phi}(\bar e_{{\phi}}(\tau)) = - \frac{1}{2} k_1 \bar e_{{\phi}}(\tau)^\top \bar e_{{\phi}}(\tau),$
    in both cases. From the definition of $V_{\phi}(\bar e_{{\phi}})$ and Rayleigh theorem \cite[Theorem 4.2.2]{Horn}, we have $\frac{1}{2} \lambda_{\min}(P_{\phi}) \lvert \bar e_{{\phi}}(\tau) \rvert^2 \leq V_{\phi}(\bar e_{{\phi}}(\tau)) \leq \frac{1}{2} \lambda_{\max}(P_{\phi}) \lvert \bar e_{{\phi}}(\tau)\rvert^2,$
	which gives $\lvert \bar e_{{\phi}}(\tau)\rvert^2 \geq \frac{2 V_{\phi}(\bar e_{{\phi}}(\tau))}{\lambda_{\max}(P_{\phi})}.$
	Consequently, given $\gamma_{\phi} = k_1 \frac{1}{\lambda_{\max}(P_{\phi})}$, we obtain
	\begin{align}\label{cond1}
		\dot V_{\phi}(\bar e_{{\phi}}(\tau)) \leq -\gamma_{\phi} V_{\phi}(\bar e_{{\phi}}(\tau)).
	\end{align}
	
	Now, let $\phi, \hat \phi \in \mathcal{P}$ be two consecutive modes and $        V_{\phi}(\bar e_{{\phi}}(t_l^+)):= \frac{1}{2} \bar e_{{\phi}}(t_l^+)^\top P_{\phi} \bar e_{{\phi}}(t_l^+)$, $V_{\hat \phi}(\bar e_{{\hat \phi}}(t_l^-)) := \frac{1}{2} \bar e_{{\hat \phi}}(t_l^-)^\top P_{\hat \phi} \bar e_{{\hat \phi}}(t_l^-).$
    It follows from \eqref{edge_switch} that for any $t_l$, $\lvert \bar e_{{\phi}}(t_l^+) \rvert^2 \leq 2\vert \Xi_{\phi, \hat \phi} \rvert^2 \lvert \bar e_{{\hat \phi}}(t_l^-) \rvert^2 + 2\lvert \Phi_l\rvert^2.$
	Noting that $\vert \Xi_{\phi, \hat \phi} \rvert \equiv 1$, because it is a submatrix of an identity matrix, and the singular values of an identity matrix are all 1, meaning that the largest singular value is also 1,
$		V_{\phi}(\bar e_{{\phi}}(t_l^+)) \leq \frac{1}{2}\lvert P_{\phi} \rvert \lvert \bar e_{{\phi}}(t_l^+) \rvert^2 \leq \lvert P_{\phi} \rvert \lvert \bar e_{{\hat \phi}}(t_l^-) \rvert^2 + \lvert P_{\phi} \rvert \lvert \Phi_l\rvert^2.$
	Since $\lvert \bar e_{{\phi}}(t_l^-)\rvert^2 \leq \frac{2 V_{\phi}(\bar e_{{\phi}}(t_l^-))}{\lambda_{\min}(P_{\hat \phi})},$
    we have
$		\lvert P_{\phi} \rvert \lvert \bar e_{{\hat \phi}}(t_l^-) \rvert^2 + \lvert P_{\phi} \rvert \lvert \Phi_l\rvert^2 \leq 2\lvert P_{\phi} \rvert \frac{V_{\phi}(\bar e_{{\phi}}(t_l^-))}{\lambda_{\min}(P_{\hat \phi})} + \lvert P_{\phi} \rvert \lvert \Phi_l\rvert^2
		\leq 2 \frac{\lambda_{\max}(P_{\phi})}{\lambda_{\min}(P_{\hat \phi})} V_{\hat \phi}(\bar e_{{\hat \phi}}(t_l^-)) + \lambda_{\max}(P_{\phi}) \lvert \Phi_l\rvert^2 $
	which gives  
	\begin{align}\label{cond2}
		V_{\phi}(\bar e_{{\phi}}(t_l^+)) \leq \Omega_{\phi , \hat \phi} V_{\hat \phi}(\bar e_{{\hat \phi}}(t_l^-)) + \Theta,
	\end{align}
	where $\Omega_{\phi , \hat \phi} = 1 + 2 \frac{\lambda_{\max}(P_{\phi})}{\lambda_{\min}(P_{\hat \phi})}$ and $\Theta = \lambda_{\max}(P_{\phi})\lvert \Phi_l\rvert^2$. Since $\lvert \Phi_l \rvert^2$ is uniformly bounded, the term $\Theta$ is finite. Hence, the Lyapunov function at switching instants is upper bounded.
    
	From \eqref{cond1}, \eqref{cond2}, and invoking Theorem \ref{thmxue}, it follows that the origin of system \eqref{dyn_edge1} is uniformly practically stable if the switching $\sigma$ admits an average dwell-time that satisfies \eqref{cond}.
\end{proof}

\section{Numerical Example}\label{section6}
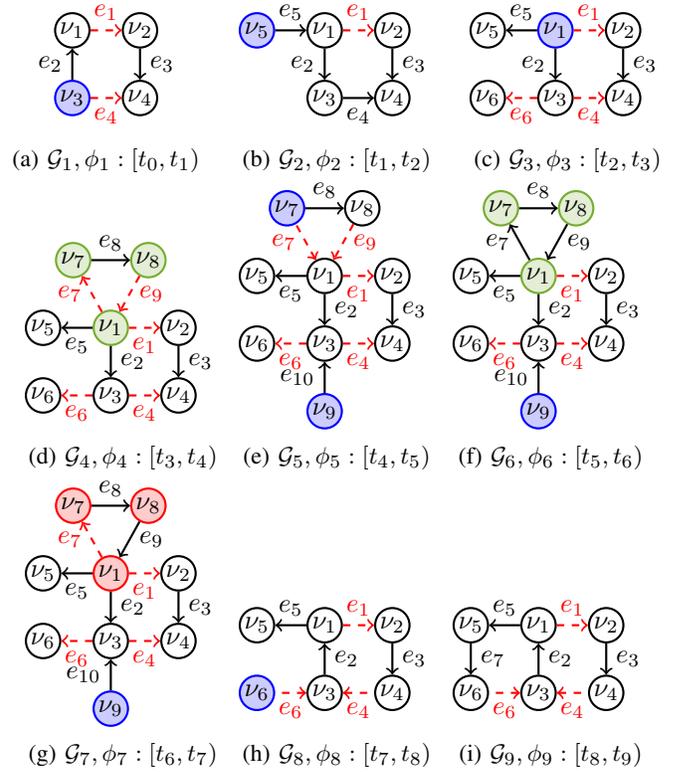
\begin{figure}[t!]\centering
\begin{subfigure}[b]{0.15\textwidth}\centering
\begin{tikzpicture}[node distance={9mm}, thick,main/.style = {draw, circle}] 
	\tikzset{mynode/.style={circle,draw,minimum size=13pt,inner sep=0pt,thick},}
		\node[mynode] (1) {$\nu_1$}; 
		\node[mynode] (2) [right  of=1] {$\nu_2$};
		\node[mynode, fill=blue, fill opacity=0.2, draw=blue, text opacity=1] (3) [below of=1] {$\nu_3$};
		\node[mynode] (4) [below of=2] {$\nu_4$};
		\draw[->, color=red, dash pattern=on 1mm off 1mm] (1) -- node[midway, above] {$e_1$}(2);
		\draw[<-] (1) -- node[midway, left] {$e_2$}(3);
		\draw[->] (2) -- node[midway, right] {$e_3$}(4);
		\draw[->, color=red, dash pattern=on 1mm off 1mm] (3) -- node[midway, below] {$e_4$}(4);
	\end{tikzpicture}
\caption{$\mathcal{G}_1, \phi_1:[t_0, t_1)$}
     \end{subfigure}
\hfill
\begin{subfigure}[b]{0.15\textwidth}
			\centering
\begin{tikzpicture}[node distance={9mm}, thick,main/.style = {draw, circle}] 
	\tikzset{mynode/.style={circle,draw,minimum size=13pt,inner sep=0pt,thick},}
		\node[mynode] (1) {$\nu_1$}; 
		\node[mynode] (2) [right  of=1] {$\nu_2$};
		\node[mynode] (3) [below of=1] {$\nu_3$};
		\node[mynode] (4) [below of=2] {$\nu_4$};
		\node[mynode, fill=blue, fill opacity=0.2, draw=blue, text opacity=1] (5) [left of=1] {$\nu_5$};
		\draw[->, color=red, dash pattern=on 1mm off 1mm] (1) -- node[midway, above] {$e_1$}(2);
		\draw[->] (1) -- node[midway, left] {$e_2$}(3);
		\draw[->] (2) -- node[midway, right] {$e_3$}(4);
		\draw[->] (3) -- node[midway, below] {$e_4$}(4);
		\draw[<-] (1) -- node[midway, above] {$e_5$}(5);
	\end{tikzpicture}
\caption{$\mathcal{G}_2, \phi_2:[t_1, t_2)$}
\end{subfigure}
\hfill
\begin{subfigure}[b]{0.15\textwidth}
	\centering
\begin{tikzpicture}[node distance={9mm}, thick,main/.style = {draw, circle}] 
	\tikzset{mynode/.style={circle,draw,minimum size=13pt,inner sep=0pt,thick},}
		\node[mynode, fill=blue, fill opacity=0.2, draw=blue, text opacity=1] (1) {$\nu_1$}; 
		\node[mynode] (2) [right  of=1] {$\nu_2$};
		\node[mynode] (3) [below of=1] {$\nu_3$};
		\node[mynode] (4) [below of=2] {$\nu_4$};
		\node[mynode] (5) [left of=1] {$\nu_5$};
		\node[mynode] (6) [left of=3] {$\nu_6$};
    \draw[red, dash pattern=on 1mm off 1mm][->](1) -- node[midway, above] {$e_1$}(2);
	\draw[->] (1) -- node[midway, left] {$e_2$}(3);
	\draw[->] (2) -- node[midway, right] {$e_3$}(4);
	\draw[->, color=red, dash pattern=on 1mm off 1mm] (3) -- node[midway, below] {$e_4$}(4);
	\draw[->] (1) -- node[midway, above] {$e_5$}(5);
	\draw[->, color=red, dash pattern=on 1mm off 1mm] (3) -- node[midway, below] {$e_6$}(6);
	\end{tikzpicture}
\caption{$\mathcal{G}_3, \phi_3:[t_2, t_3)$}
\end{subfigure}
\begin{subfigure}[b]{0.15\textwidth}
	\centering
\begin{tikzpicture}[node distance={9mm}, thick,main/.style = {draw, circle}] 
	\tikzset{mynode/.style={circle,draw,minimum size=13pt,inner sep=0pt,thick},}
    \definecolor{pastelgreen}{rgb}{0.4660 0.6740 0.1880}
	\node[mynode, fill=pastelgreen, fill opacity=0.2, draw=pastelgreen, text opacity=1] (1) {$\nu_1$}; 
	\node[mynode] (2) [right  of=1] {$\nu_2$};
	\node[mynode] (3) [below of=1] {$\nu_3$};
	\node[mynode] (4) [below of=2] {$\nu_4$};
	\node[mynode] (5) [left of=1] {$\nu_5$};
	\node[mynode] (6) [left of=3] {$\nu_6$};
    \node[mynode, fill=pastelgreen, fill opacity=0.2, draw=pastelgreen, text opacity=1] (7) [above of=1, xshift=-0.5cm] {$\nu_7$};
	\node[mynode, fill=pastelgreen, fill opacity=0.2, draw=pastelgreen, text opacity=1] (8) [above of=1, xshift=0.5cm] {$\nu_8$};
    \draw[red, dash pattern=on 1mm off 1mm][->](1) -- node[midway, below] {$e_1$}(2);
	\draw[->] (1) -- node[midway, right] {$e_2$}(3);
	\draw[->] (2) -- node[midway, right] {$e_3$}(4);
	\draw[->, color=red, dash pattern=on 1mm off 1mm] (3) -- node[midway, below] {$e_4$}(4);
	\draw[->] (1) -- node[midway, below] {$e_5$}(5);
	\draw[->, color=red, dash pattern=on 1mm off 1mm] (3) -- node[midway, below] {$e_6$}(6);
    \draw[->] (7) -- node[midway, above] {$e_8$}(8);
    \draw[->, color=red, dash pattern=on 1mm off 1mm] (8) -- node[midway, right] {$e_9$}(1);
    \draw[->, color=red, dash pattern=on 1mm off 1mm] (1) -- node[midway, left] {$e_7$}(7);
	\end{tikzpicture}
	\caption{$\mathcal{G}_4, \phi_4:[t_3, t_4)$}
\end{subfigure}
\begin{subfigure}[b]{0.15\textwidth}
	\centering
\begin{tikzpicture}[node distance={9mm}, thick,main/.style = {draw, circle}] 
	\tikzset{mynode/.style={circle,draw,minimum size=13pt,inner sep=0pt,thick},}
    \definecolor{pastelgreen}{rgb}{0.4660 0.6740 0.1880}
	\node[mynode] (1) {$\nu_1$}; 
	\node[mynode] (2) [right  of=1] {$\nu_2$};
	\node[mynode] (3) [below of=1] {$\nu_3$};
	\node[mynode] (4) [below of=2] {$\nu_4$};
	\node[mynode] (5) [left of=1] {$\nu_5$};
	\node[mynode] (6) [left of=3] {$\nu_6$};
    \node[mynode, fill=blue, fill opacity=0.2, draw=blue, text opacity=1] (7) [above of=1, xshift=-0.5cm] {$\nu_7$};
	\node[mynode] (8) [above of=1, xshift=0.5cm] {$\nu_8$};
	\node[mynode, fill=blue, fill opacity=0.2, draw=blue, text opacity=1] (9) [below of=3] {$\nu_9$};
    \draw[red, dash pattern=on 1mm off 1mm][->](1) -- node[midway, below] {$e_1$}(2);
	\draw[->] (1) -- node[midway, right] {$e_2$}(3);
	\draw[->] (2) -- node[midway, right] {$e_3$}(4);
	\draw[red, dash pattern=on 1mm off 1mm][->] (3) -- node[midway, below] {$e_4$}(4);
	\draw[->] (1) -- node[midway, below] {$e_5$}(5);
	\draw[->, color=red, dash pattern=on 1mm off 1mm] (3) -- node[midway, below] {$e_6$}(6);
    \draw[->] (7) -- node[midway, above] {$e_8$}(8);
    \draw[->, color=red, dash pattern=on 1mm off 1mm] (8) -- node[midway, right] {$e_9$}(1);
    \draw[<-, color=red, dash pattern=on 1mm off 1mm] (1) -- node[midway, left] {$e_7$}(7);
    \draw[->] (9) -- node[midway, left] {$e_{10}$}(3);
	\end{tikzpicture}
	\caption{$\mathcal{G}_5, \phi_5:[t_4, t_5)$}
\end{subfigure}
\begin{subfigure}[b]{0.15\textwidth}
	\centering
	\begin{tikzpicture}[node distance={9mm}, thick,main/.style = {draw, circle}] 
     \definecolor{pastelgreen}{rgb}{0.4660 0.6740 0.1880}
\tikzset{mynode/.style={circle,draw,minimum size=13pt,inner sep=0pt,thick},}
	\node[mynode, fill=pastelgreen, fill opacity=0.2, draw=pastelgreen, text opacity=1] (1) {$\nu_1$}; 
	\node[mynode] (2) [right  of=1] {$\nu_2$};
	\node[mynode] (3) [below of=1] {$\nu_3$};
	\node[mynode] (4) [below of=2] {$\nu_4$};
	\node[mynode] (5) [left of=1] {$\nu_5$};
	\node[mynode] (6) [left of=3] {$\nu_6$};
    \node[mynode, fill=pastelgreen, fill opacity=0.2, draw=pastelgreen, text opacity=1] (7) [above of=1, xshift=-0.5cm] {$\nu_7$};
	\node[mynode, fill=pastelgreen, fill opacity=0.2, draw=pastelgreen, text opacity=1] (8) [above of=1, xshift=0.5cm] {$\nu_8$}; %
	\node[mynode, fill=blue, fill opacity=0.2, draw=blue, text opacity=1] (9) [below of=3] {$\nu_9$};
    \draw[red, dash pattern=on 1mm off 1mm][->](1) -- node[midway, below] {$e_1$}(2);
	\draw[->] (1) -- node[midway, right] {$e_2$}(3);
	\draw[->] (2) -- node[midway, right] {$e_3$}(4);
	\draw[->, color=red, dash pattern=on 1mm off 1mm] (3) -- node[midway, below] {$e_4$}(4);
	\draw[->] (1) -- node[midway, below] {$e_5$}(5);
	\draw[->, color=red, dash pattern=on 1mm off 1mm] (3) -- node[midway, below] {$e_6$}(6);
    \draw[->] (7) -- node[midway, above] {$e_8$}(8);
    \draw[->] (8) -- node[midway, right] {$e_9$}(1);
    \draw[->] (1) -- node[midway, left] {$e_7$}(7);
    \draw[->] (9) -- node[midway, left] {$e_{10}$}(3);
	\end{tikzpicture}
	\caption{$\mathcal{G}_6, \phi_6:[t_5, t_6)$}
\end{subfigure}
\begin{subfigure}[b]{0.15\textwidth}
	\centering
	\begin{tikzpicture}[node distance={9mm}, thick,main/.style = {draw, circle}] 
\tikzset{mynode/.style={circle,draw,minimum size=13pt,inner sep=0pt,thick},}
	\node[mynode, fill=red, fill opacity=0.2, draw=red, text opacity=1] (1) {$\nu_1$}; 
	\node[mynode] (2) [right  of=1] {$\nu_2$};
	\node[mynode] (3) [below of=1] {$\nu_3$};
	\node[mynode] (4) [below of=2] {$\nu_4$};
	\node[mynode] (5) [left of=1] {$\nu_5$};
	\node[mynode] (6) [left of=3] {$\nu_6$};
    \node[mynode, fill=red, fill opacity=0.2, draw=red, text opacity=1] (7) [above of=1, xshift=-0.5cm] {$\nu_7$};
	\node[mynode, fill=red, fill opacity=0.2, draw=red, text opacity=1] (8) [above of=1, xshift=0.5cm] {$\nu_8$}; %
	\node[mynode, fill=blue, fill opacity=0.2, draw=blue, text opacity=1] (9) [below of=3] {$\nu_9$};
    \draw[red, dash pattern=on 1mm off 1mm][->](1) -- node[midway, below] {$e_1$}(2);
	\draw[->] (1) -- node[midway, right] {$e_2$}(3);
	\draw[->] (2) -- node[midway, right] {$e_3$}(4);
	\draw[->, color=red, dash pattern=on 1mm off 1mm] (3) -- node[midway, below] {$e_4$}(4);
	\draw[->] (1) -- node[midway, below] {$e_5$}(5);
	\draw[->, color=red, dash pattern=on 1mm off 1mm] (3) -- node[midway, below] {$e_6$}(6);
    \draw[->] (7) -- node[midway, above] {$e_8$}(8);
    \draw[->] (8) -- node[midway, right] {$e_9$}(1);
    \draw[->, color=red, dash pattern=on 1mm off 1mm] (1) -- node[midway, left] {$e_7$}(7);
    \draw[->] (9) -- node[midway, left] {$e_{10}$}(3);
	\end{tikzpicture}
	\caption{$\mathcal{G}_7, \phi_7:[t_6, t_7)$}
\end{subfigure}
\begin{subfigure}[b]{0.15\textwidth}
	\centering
	\begin{tikzpicture}[node distance={9mm}, thick,main/.style = {draw, circle}] 
    \tikzset{mynode/.style={circle,draw,minimum size=13pt,inner sep=0pt,thick},}
	\node[mynode] (1) {$\nu_1$}; 
	\node[mynode] (2) [right  of=1] {$\nu_2$};
	\node[mynode] (3) [below of=1] {$\nu_3$};
	\node[mynode] (4) [below of=2] {$\nu_4$};
	\node[mynode] (5) [left of=1] {$\nu_5$};
	\node[mynode, fill=blue, fill opacity=0.2, draw=blue, text opacity=1] (6) [left of=3] {$\nu_6$};
    \draw[red, dash pattern=on 1mm off 1mm][->](1) -- node[midway, above] {$e_1$}(2);
	\draw[<-] (1) -- node[midway, right] {$e_2$}(3);
	\draw[->] (2) -- node[midway, right] {$e_3$}(4);
	\draw[<-, color=red, dash pattern=on 1mm off 1mm] (3) -- node[midway, below] {$e_4$}(4);
	\draw[->] (1) -- node[midway, above] {$e_5$}(5);
	\draw[<-, color=red, dash pattern=on 1mm off 1mm] (3) -- node[midway, below] {$e_6$}(6);
	\end{tikzpicture}
	\caption{$\mathcal{G}_8, \phi_8:[t_7, t_8)$}
\end{subfigure}
\begin{subfigure}[b]{0.15\textwidth}
	\centering
	\begin{tikzpicture}[node distance={9mm}, thick,main/.style = {draw, circle}] 
    \tikzset{mynode/.style={circle,draw,minimum size=13pt,inner sep=0pt,thick},}
	\node[mynode] (1) {$\nu_1$}; 
	\node[mynode] (2) [right  of=1] {$\nu_2$};
	\node[mynode] (3) [below of=1] {$\nu_3$};
	\node[mynode] (4) [below of=2] {$\nu_4$};
	\node[mynode] (5) [left of=1] {$\nu_5$};
	\node[mynode] (6) [left of=3] {$\nu_6$};
    \draw[red, dash pattern=on 1mm off 1mm][->](1) -- node[midway, above] {$e_1$}(2);
	\draw[<-] (1) -- node[midway, right] {$e_2$}(3);
	\draw[->] (2) -- node[midway, right] {$e_3$}(4);
	\draw[<-, color=red, dash pattern=on 1mm off 1mm] (3) -- node[midway, below] {$e_4$}(4);
	\draw[->] (1) -- node[midway, above] {$e_5$}(5);
	\draw[<-, color=red, dash pattern=on 1mm off 1mm] (3) -- node[midway, below] {$e_6$}(6);
    \draw[<-] (6) -- node[midway, right] {$e_7$}(5);
	\end{tikzpicture}
	\caption{$\mathcal{G}_9, \phi_9:[t_8, t_9)$}
\end{subfigure}
  \caption{The black edges represent cooperative interactions, and the dashed red edges represent antagonistic interactions.} 
  \label{graph}
\end{figure}
We illustrate our theoretical findings by considering a signed OMAS evolving according to \eqref{FO}--\eqref{CL}, initially represented by a SB signed digraph with four agents, where $\mathcal{V}_1 = \{\nu_1, \nu_3\},\ \mathcal{V}_2 = \{\nu_2, \nu_4\}$. Over time, agents are added or removed, the signs of interconnections switch between cooperation and antagonism making the signed digraph SUB, and agents change roles, from being followers to leaders, as illustrated in Figure \ref{graph}. Let $k_1=4$ and $P_\phi$ be generated by \eqref{Lyap-eq} or \eqref{Lyap_eq_dir} for each mode $\phi \in \mathcal{P}$, depending on the considered graph structure, with $Q_{\phi} = I_{M_{\phi} \times M_{\phi}}$, $\alpha_{\phi_j}=1$, $j\leq 7$, $\alpha_{\phi_8}=0.2$ and $\alpha_{\phi_9}=0.4$. The initial conditions of the agents are $[3.5,\ 4,\ -2,\ -6.5,\ 5.5,\ -10.5,\ 3.5,\ 12,\ 5.5]$. Lastly, the switching times in seconds are \( t_1 = 1.5 \), \( t_2 = 3 \), \( t_3 = 4.5 \), \( t_4 = 6.8 \), \( t_5 = 9 \), \( t_6 = 12 \), \( t_7 = 17 \), and \( t_8 = 24 \), which satisfy \eqref{cond}.
\begin{figure*}[h!]
\centering
\subfloat[]{
\includegraphics[width=0.85\columnwidth]{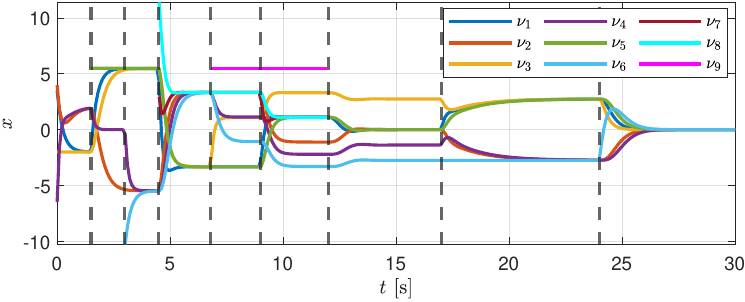}\label{fig2}}\quad
\subfloat[]{
\includegraphics[width=0.85\columnwidth]{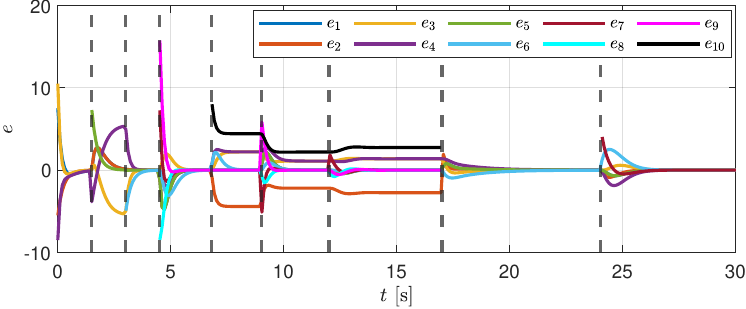}
\label{fig3}}\\[5pt]
\subfloat[]{\includegraphics[width=0.85\columnwidth]{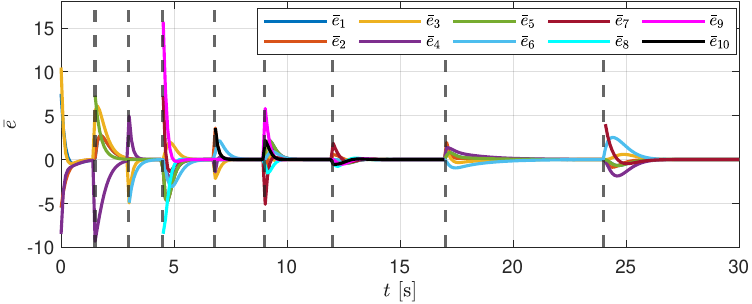}\label{fig4}}
\quad
\subfloat[]{\includegraphics[width=0.85\columnwidth]{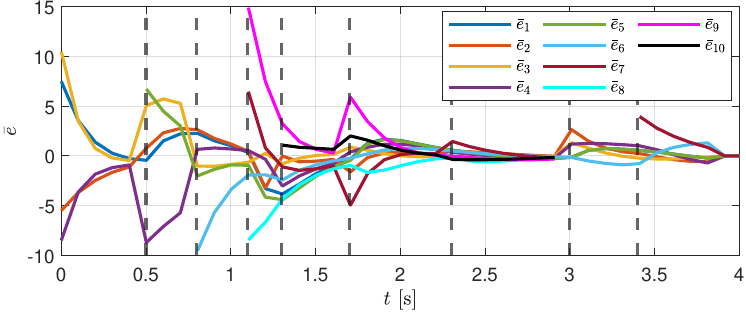}\label{fig4_1}}
\caption{ {\small Evolution of the trajectories. (a): Agents \eqref{FO}-\eqref{CL}. (b): Edges \eqref{edgestates}. (c): Synchronization errors  \eqref{dyn_edge1} when $\tau$ satisfies \eqref{cond}. (d): synchronization errors \eqref{dyn_edge1} when $\tau$ does not satisfy \eqref{cond}.}} 
\end{figure*}
The evolution of the agents' state trajectories is shown in Figure \ref{fig2}, the edge states in Figure \ref{fig3}, and the synchronization errors in Figure \ref{fig4}. 
The peaks observed in Figure \ref{fig3} correspond to the addition/removal of edges and changes in edge signs or directions. Unlike the results in \cite{CDC2025OMAS}, which consider only undirected signed graphs where edge states always converge to zero, here it is evident that the edge states do not always converge to zero. For example, when the resulting graph contains multiple leader groups, \textit{i.e.,} when the system admits more than two equilibria, the edge states remain nonzero. Finally, Figure \ref{fig4} shows that the synchronization errors converge very close to zero, when \eqref{cond} is satisfied. In contrast, in Figure \ref{fig4_1}, when $\tau$ fails to satisfy \eqref{cond}, the synchronization errors deviate noticeably from zero.

\section{Conclusion}\label{section7}
In this paper, we addressed the synchronization problem of signed OMAS, where both cooperative and antagonistic interactions are present, and which potentially contain one or multiple leader groups. 
We extended the Lyapunov equation-based characterization of the Laplacians of all-cooperative networks with one zero eigenvalue to signed edge Laplacians with multiple zero eigenvalues. 
Finally, we analyzed the stability properties of the signed OMAS and investigated the different synchronization scenarios, from standard consensus to more complex scenarios of bipartite containment. Further research aims to extend these results to the case of constrained systems and discrete-time settings.

\section*{References}
\bibliographystyle{IEEEtran}
\bibliography{references_pelin}

\end{document}